\documentclass[12pt]{amsart}

\usepackage{graphicx, amssymb}

\usepackage{axodraw4j}
\usepackage{pstricks}
\usepackage{color}

\newtheorem{thm}{Theorem}
\newtheorem{lemma}[thm]{Lemma}
\newtheorem{prop}[thm]{Proposition}
\newtheorem{remark}[thm]{Remark}
\newtheorem{cor}[thm]{Corollary}

\theoremstyle{definition}
\newtheorem{definition}[thm]{Definition}
\newtheorem{example}[thm]{Example}

\newcommand{\C}{\mathbb C}
\newcommand{\Q}{\mathbb Q}

\newcommand{\x}{\mathsf{x}}
\newcommand{\gr}{\mathrm{gr}}
\newcommand{\LL}{\mathcal{L}}

\newcommand{\lc}{\langle \langle}
\newcommand{\rc}{\rangle \rangle}

\title[Weight of Feynman graphs]{Spanning forest polynomials and the transcendental weight of  Feynman graphs}

\author{Francis Brown and Karen Yeats}
\thanks{Karen Yeats is supported by an NSERC discovery grant}

\begin{document}
\begin{abstract}
  We give combinatorial criteria for predicting the transcendental weight of Feynman integrals of certain graphs in $\phi^4$ theory.
   By studying spanning forest polynomials, we obtain 
    operations on graphs which are weight-preserving, and a list of   subgraphs which induce a drop in the transcendental weight.
\end{abstract} 

\maketitle

\vspace{-0.5in}
\section{Introduction}

It is well-known  since the work of Broadhurst and  Kreimer  that single-scale massless Feynman integral calculations in perturbative quantum field theories
give rise empirically to multiple zeta values. In particular,  there is a map from primitive graphs in $\phi^4$ theory at low loop orders to linear combinations of multiple zeta values.
Currently there is no way to predict this map without intensive numerical analysis.



In this paper, we consider  the most simple invariant of multiple zeta values: their transcendental  weight. Conjecturally, there should exist no linear relations between 
multiple zetas of different weights, and hence one expects there to  be a grading on the ring of MZVs.
%
On the other hand, the perturbative expansion in massless $\phi^4$ theory is  also graded by the number of loops. The surprising fact is that these gradings do not quite coincide.
In the generic case, the transcendental weight of a   graph is equal to twice its loop number minus 3. This is the case for the left and middle graphs in the examples below:
%
%
%
%
%
%
%
\begin{center}
\fcolorbox{white}{white}{
  \begin{picture}(322,94) (15,-120)
    \SetWidth{1.0}
    \SetColor{Black}
    \Text(48,-135)[lb]{\Black{$6\zeta(3)$}}
    \Text(161,-135)[lb]{\Black{$20\zeta(5)$}}
    \Text(280,-135)[lb]{\Black{$36\zeta(3)^2$}}
    \Text(-2,-35)[lb]{\Black{$3$ loops:}}
    \Text(111,-35)[lb]{\Black{$4$ loops:}}
    \Text(230,-35)[lb]{\Black{$5$ loops:}}
    \Arc(60,-74)(44,270,630)
    \Arc(174,-74)(44,270,630)
    \Arc(292,-74)(44,270,630)
    \Vertex(60,-30){2.828}
    \Vertex(174,-30){2.828}
    \Vertex(130,-74){2.828}
    \Vertex(218,-74){2.828}
    \Vertex(174,-118){2.828}
    \Vertex(94,-101){2.828}
    \Vertex(25,-101){2.828}
    \Line(60,-30)(60,-74)
    \Line(60,-74)(94,-101)
    \Line(60,-74)(25,-101)
    \Line(174,-30)(174,-116)
    \Line(132,-74)(218,-74)
    \Vertex(60,-74){2.828}
    \Vertex(174,-74){2.828}
    \Vertex(292,-30){2.828}
    \Vertex(292,-58){2.828}
    \Vertex(271,-82){2.828}
    \Vertex(315,-82){2.828}
    \Vertex(315,-112){2.828}
    \Vertex(271,-112){2.828}
    \Line(292,-30)(292,-58)
    \Line(292,-58)(271,-82)
    \Line(271,-82)(271,-112)
    \Line(292,-58)(315,-82)
    \Line(315,-82)(315,-112)
    \Line(315,-82)(271,-112)
    \Line[dash,dashsize=22](271,-82)(315,-112)
  \end{picture}
}
\end{center}

\vspace{0.2in}
But the   non-planar graph on the  right has 5 loops 
and hence its expected  weight should be  $2\times 5-3=7$; yet it has weight 6. In other words,   a \emph{weight drop} can occur.
The goal of this paper is to understand     combinatorially why  such a weight drop arises.

Our main result  describes some operations on graphs in $\phi^4$ theory under which the weight is preserved. As a corollary, we produce some
infinite families of graphs which are of maximal weight, and other families  which have  a weight drop.  These two classes  should contribute to $\phi^4$ theory in a 
quite different way.  In order to obtain  physical predictions in practice,  one must  sum a large number of diagrams in a given quantum field theory  at each loop order,  and it is known that not all graphs contribute equally
to the final sum. Indeed, some can even be discarded altogether.  We hope that the notion of  weight drop  may shed some light on this phenomenon.
 
 \begin{remark} It is currently not known whether all primitive graphs in $\phi^4$ theory evaluate to multiple zetas or not, but the residues are always periods in the sense
 of \cite{Ko-Za}. Therefore the  general conjectural picture is that there should exist a  large  pro-algebraic  (`motivic') Galois group which acts on the set of all periods, and should in particular equip   the perturbative expansion of a quantum field theory with a lot of extra structure. The notion of weight is the first non-trivial piece of information that such a theory would provide.
 \end{remark}
 
\subsection{Outline} Let $G$ be a primitive graph in $\phi^4$ theory with $e_G$ edges. Its residue is defined by the formula
$$I_G = \int  {1\over \Psi_G^2} \prod_{i=1}^{e_G} d\alpha_i \delta(\alpha_{\kappa}=1)\ ,$$
where $\alpha_i$ is the Schwinger coordinate of each edge, and $I_G$ does not depend on the choice of  $\kappa$.
The graph polynomial $\Psi_G$ is  defined by 
$$\Psi_G = \sum_{T\subseteq G} \prod_{e\notin T} \alpha_e$$
where $T$ ranges over all spanning trees of $G$.  In order to understand the  integral $I_G$, one is naturally led \cite{Brbig} to consider auxilliary (or `Dodgson') polynomials 
$\Psi^{I,J}_{G,K}$,   where $I,J,K$ are subsets of edges of $G$ satisfying $|I|=|J|$.
 In the first part of this paper, we introduce \emph{spanning forest polynomials}  $\Phi^{P}_G$ associated to 
 any partition $P$ of a subset  of vertices of $G$.  These are sums over families of trees whose leaves   contain the vertices in each partition of $P$.
 We show that every polynomial $\Psi^{I,J}_{G,K}$ can be written as a linear combination of $\Phi^{P}_G$, which in particular gives a formula for the signs in $\Psi^{I,J}_{G,K}$.

Next, in \S\ref{sectId} we study some algebraic identities between spanning forest polynomials, and
give a universal formula for the graph polynomial of any 3-vertex connected graph  $G$ in terms of the $\Phi^{P}_G$. The graph polynomial for any such graph
is the graph polynomial of the single graph $M(1,1,1)$ which has 6 edges, which are decorated by spanning forest polynomials. Thus to prove a statement about
any 3-vertex connected graph, it suffices to prove it in this single case.

In \S\ref{sectHyper} we recall some elementary properties of hyperlogarithms  and give a sufficient condition for a graph to have a weight drop. This can be phrased in terms
of higher graph invariants \cite{Brbig}, which in some special cases can be computed in terms of spanning forest polynomials.
 Using the properties of the $\Phi^P_G$ proved previously, we show that any 2-vertex connected graph always has a weight drop. We then show 
 that the operation of  splitting a triangle and moving some of its outer edges preserves weights.
 These two results alone suffice to explain almost all of the known weight-drop graphs in $\phi^4$ theory up to 8 loops.

In \S\ref{sect3vw}, we  seek a classification of weight-preserving, or weight-dropping  operations.  Using the universal formula for 3-vertex  connected graphs, it suffices to write down
all the local minors of 3-vertex connected graphs at $k$ edges, for small values of $k$, and compute their graph polynomials.
From this we deduce some new families of weight-preserving operations which are not attainable by splitting triangles.
\\

Very many thanks to S. Bloch, D. Broadhurst,  D. Kreimer, and  O. Schnetz for discussions and enthusiasm.
\\

\begin{remark}
One way to circumvent the transcendence conjectures for periods is to replace them with  mixed Hodge structures, which was  initiated in  \cite{bek}. The question
one can then ask, following (loc cit), is where in the  weight filtration does  the differential form defining the period sit? In this context, the weight always makes sense, but it is currently not known how to carry this out save for a  few examples 
of graphs. Therefore, in the face of the geometric difficulties of this problem, we have instead focused on the combinatorial aspects of the weight drop, which seems to be a necessary prerequisite before trying to tackle the Hodge side. It came as a surprise to us quite how intricate this first step already is.
\end{remark}




\subsection{Background}
Let $G$ be a connected multigraph, with self-loops\footnote{or \emph{tadpoles}} allowed. Let $E(G)$, $V(G)$ denote the set of edges and vertices of $G$, and let $e_G=|E(G)|$, $v_G = |V(G)|$.
 To each edge $e$ of $G$, we  associate a Schwinger parameter $\alpha_e$. The graph polynomial of $G$ is defined by
 $$\Psi_G= \sum_{T\subseteq G} \prod_{e\notin E(T)} \alpha_e\in \mathbb{Z}[\alpha_e, e\in E(G)]\ , $$
 where the sum is over all spanning trees $T$ of $G$.

\begin{definition} Choose an orientation on the edges of $G$, and for every edge $e$ and vertex $v$ of $G$, define an incidence matrix:
$$(\mathcal{E}_G)_{e,v} = \left\{
                           \begin{array}{ll}
                             1, & \hbox{if the edge } e \hbox{ begins at } v, \\
                             -1, & \hbox{if the edge } e  \hbox{ ends at } v ,\\
                             0, & \hbox{otherwise}.
                           \end{array}
                         \right.
 $$
Let $A$ be the diagonal matrix with entries  $\alpha_e$, for $e \in E(G)$, and set
$$\widetilde{M}_G=\left(
  \begin{array}{c|c}
    A  & \mathcal{E}_G  \\
    \hline
  {-}\mathcal{E}_G^T&  0  \\
  \end{array}
\right)
$$
where the first $e_G$ rows (resp. columns) are indexed by the set of edges of $G$, and the remaining $v_G$ rows (resp. columns) are indexed by the set of vertices of $G$, in some order.
The matrix $\widetilde{M}_G$ has zero determinant.
\end{definition}

\begin{definition} Choose any vertex of $G$ and let $M_G$ denote the minor  obtained by  deleting the corresponding row and column of $\widetilde{M}_G$.  
\end{definition}
The matrix $M_G$ is not well-defined, but one can show that 
\[
    \Psi_G = \det (M_G)
  \]
  is the graph polynomial of $G$. This motivates the following:

\begin{definition} Let $I,J, K$ be subsets of the set of edges of $G$ which satisfy $|I|=|J|$. Let  
$M_G(I,J)_K$ denote the matrix obtained from $M_G$ by removing the rows (resp. columns) indexed by the set $I$ (resp. $J$) and setting $\alpha_e=0$ for all $e\in K$.
Set 
$$\Psi_{G,K}^{I,J}=\det M_G(I,J)_K\ . $$
\end{definition}

It is clear that $\Psi_{G,\emptyset}^{\emptyset,\emptyset} = \Psi_G$.  The polynomials $\Psi^{I,J}_{G,K}$ are well-defined up to sign.
The following results are proved in \cite{Brbig}.
\begin{prop}\label{contract} Let $e\in E(G)$ such that $e\notin I \cup J \cup K$. Let $G\backslash e$ denote the graph obtained by deleting the edge $e$ (and removing any isolated vertices) and let 
$G/ e$ denote the graph obtained by contracting $e$ (i.e., deleting $e$ and identifying its endpoints).\footnote{We require that the contraction of tadpoles be zero} Then
 \begin{eqnarray} \Psi_{G\backslash e, K}^{I,J} &= & \Psi_{G, K}^{I\cup e, J\cup e} \nonumber \\
   \Psi_{G/ e, K}^{I,J} &= & \Psi_{G, K\cup e}^{I, J} \nonumber 
 \end{eqnarray} 
 Since  $\Psi_{G,K}^{I,J}$ is linear in the Schwinger parameters this implies that
 $$\Psi_{G,K}^{I,J} = \Psi_{G\backslash e, K}^{I, J} \,\alpha_e + \Psi_{G/e, K}^{I, J}.$$
\end{prop}

\begin{cor}\label{IJdisj} Let $G$, $I,J,K$ be as above. Then 
  $$\Psi^{I,J}_{G,K} = \Psi^{I',J'}_{G',\emptyset}$$
  where $G' = G\backslash (I\cap J)/ (K\backslash (I\cap J))$, and $I'\cap J'=\emptyset$. In other words, by passing to a minor of $G$, we can  assume that $I\cap J=K=\emptyset$.
  \end{cor}

Now let $U\subset E(G)$ be a set of $|U|=h_1(G)=e_G-v_G+1$ edges.  Define $\mathcal{E}_G(U)$ to be a square $ (v_G-1)\times (v_G-1)$ matrix obtained from $\mathcal{E}_G$ by removing  the rows $U$, and any one column corresponding  to a vertex.  The matrix-tree theorem  states that
$$\det \mathcal{E}_G(U) \in \{0,\pm 1\}\ ,$$
and is non-zero if and only if $U$ is a spanning tree of $G$. 
\begin{prop}\label{E}
 Suppose that $I\cap J=\emptyset$. Then 
 $$\Psi_{G,\emptyset}^{I,J} =\sum_{U \subset G\backslash (I\cup J)} \prod_{u \notin U} \alpha_u \det (\mathcal{E}_G(U\cup I)) \det (\mathcal{E}_G(U \cup J))$$
  where $U$ ranges over all subgraphs of $G\backslash (I\cup J)$  which have the property that 
  %
  %
   %
   $U\cup I$ and $U\cup J$ are both spanning trees in $G$. 
\end{prop}

The first goal of this paper is to provide combinatorial interpretations of the $\Psi^{I,J}_{G,K}$.  The objects we will use to this end are spanning forests, that is, subgraphs of $G$ which contain all vertices of $G$ and are  disjoint unions of trees.

\begin{definition}
  Let $P=P_1\cup \ldots \cup P_k$ be a set partition of a subset of the vertices of $G$.  Define
  \[
    \Phi_G^P = \sum_F\prod_{e \not\in F} \alpha_e
  \]
  where the sum runs over spanning forests $F=T_1\cup \ldots \cup T_k$ where each tree $T_i$ of $F$ 
  contains the  vertices in $P_i$ and no other vertices of $P$, i.e., $V(T_i) \supseteq P_i$ and $V(T_i)\cap P_j =\emptyset $ for  $j\neq i$.  Trees consisting of a single vertex are permitted.
 Call $\Phi_G^P$ a \emph{spanning forest polynomial} of $G$. 
\end{definition}

Graphically we will represent $\Phi_G^P$ by associating a colour to each part of $P$ and drawing $G$ with the vertices in $P$ coloured accordingly.
For example, let $G$ be the wheel with three spokes, labelled as illustrated.
\[
\includegraphics{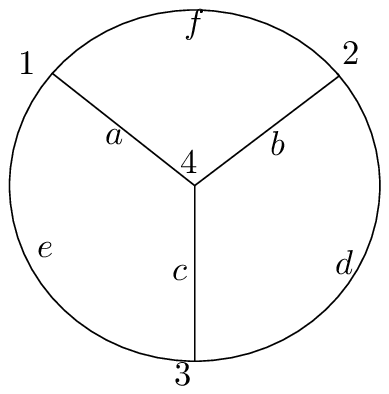}
\]
Let $P = \{1\},\{2,4\}$.  Then
\[
  \Phi_G^P = af(be + cd + ce + de) 
\]
Graphically we represent this
\[
\includegraphics{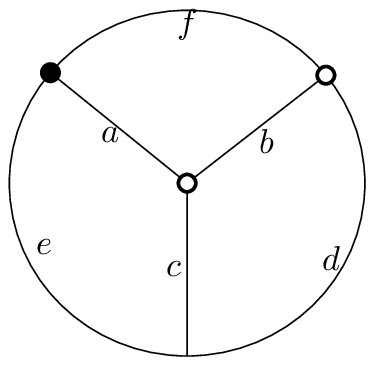}
\]

Spanning forest polynomials are well behaved under contraction and deletion.  The following propositions can be most easily understood simply by drawing each $\Phi^P_G$.

\begin{prop}\label{propdiffparts}
  Let $e$ be an edge variable of $G$ and let $P$ be a set partition of some vertices of $G$.  Then
  \[
  \Phi^P_G = \begin{cases} e \Phi^P_{G \backslash e} & \text{if the ends of $e$ are in different parts of $P$} \\ e \Phi^P_{G \backslash e} + \Phi^{P / e}_{G / e} & \text{otherwise}\end{cases}
  \]
  Where $P/e$ is the set partition made from $P$ by identifying the two ends of $e$ should they appear in $P$.
\end{prop}

\begin{proof}
  If the ends of $e$ are in different parts of $P$ then the edge $e$ must not appear in the spanning forest polynomial.  Thus factoring out $e$ is equivalent to cutting $e$.  This leaves a spanning forest of $G\backslash e$ compatible with $P$.
  The second case follows as in the graph polynomial case.
\end{proof}

\begin{prop}
  Let $v$ and $w$ be the two ends of $e$.  Assume that $v$ and $w$ are not in different parts of $P$.  Then 
  \[
  \Phi^{P / e}_{G / e} = 
  \begin{cases}
    \displaystyle\sum_{p \text{ part of } P} \sum_{\substack{p_1 \cup p_2 = p \\ p_1 \cap p_2 = \emptyset}}\Phi^{P\backslash p , p_1 \cup \{v\}, p_2 \cup \{w\}}_{G\backslash e} 
    & \text{$v,w \not\in P$} \\
    \displaystyle\sum_{\substack{p_1 \cup p_2 = p \cap (V(G) \backslash \{v,w\}) \\ p_1 \cap p_2 = \emptyset}}\Phi^{P \backslash p, p_1 \cup \{v\}, p_2 \cup \{w\}}_{G\backslash e} 
    & \text{$v$ in part $p$ of $P$}
  \end{cases}
  \]
  where $P \backslash p$ means the partition consisting of the parts of $P$ other than $p$, and $P, q$ means the partition with the parts of $P$ and with the part $q$.
\end{prop}

Note that $p_1\cup p_2$ and $p_2\cup p_1$ are counted as two different terms of the sums provided $p_1 \neq p_2$.

\begin{proof}
  Let $u$ be the vertex $v\sim w$ in $G/e$.

  Suppose $v$ and $w$ are not in $P$.  Then $u$ is not in $P/e$, so $P/e = P$ and $u$ may belong to any tree of a spanning forest of $G/e$ giving
  \[
    \Phi^{P/e}_{G/e}= \sum_{p \text{ part of } P} \Phi^{P\backslash p, p\cup \{u\}}_{G/e} .
  \]
  Any spanning forest corresponding to the partition
  \[
  P\backslash p, p\cup \{u\}
  \] in $G/e$, is also a spanning forest with one more tree in $G$.  This extra tree splits $p \cup \{u\}$ into $p_1 \cup \{v\}$ and $p_2 \cup \{w\}$.  Specifically,
  \[
    \Phi^{P\backslash p, p\cup \{u\}}_{G/e} =  \sum_{\substack{p_1 \cup p_2 = p \\ p_1 \cap p_2 = \emptyset}}\Phi^{P\backslash p , p_1 \cup \{v\}, p_2 \cup \{w\}}_{G\backslash e} ,
  \]
  giving the formula for $v,w\not\in P$.

  Now suppose $v$ is in part $p$ of $P$, and so $w$ is either not in $P$ or is also in $p$.   Then 
\[
  P/e = P\backslash p, p'\cup \{u\}
\] where $p' = p \cap (V(G) \backslash \{v, w\})$.   Any spanning forest in $G/e$ corresponding to this set partition is again also a spanning forest with one more tree in $G$.  This extra tree splits $p' \cup \{u\}$ into $p_1 \cup \{v\}$ and $p_2 \cup \{w\}$ where $p_1 \cup p_2 = p'$, leading, as above, to the desired sum decomposition.
\end{proof}

\section{Signs in Dodgson polynomials}

 The $\Psi_{G,K}^{I,J}$ can be expanded in terms of spanning forest polynomials.  This expansion provides a simple combinatorial explanation for the signs of the monomials of the $\Psi_{G,K}^{I,J}$.

In view of Corollary \ref{IJdisj} it suffices to consider $\Psi_{G,K}^{I,J}$ with $K = \emptyset$ and $I\cap J = \emptyset$.  Thus we will suppress $K = \emptyset$ from the notation.

\begin{prop}\label{exists}
  Let $I$ and $J$ be sets of edges of $G$ with $|I| = |J|$ and $I \cap J = \varnothing$.  Then we can write
  \[
    \Psi_G^{I,J} = \sum_{k} f_k \Phi_G^{P_k}
  \]
  where the sum runs over partitions of $V(I\cup J)$ and $f_k \in \{-1, 0, 1\}$.
\end{prop}

\begin{proof} 
  Take a particular monomial $m$ in $\Psi_G^{I,J}$.   Let  $M$ denote the set of edges in $m$, and $N$  the complementary set of edges in $G\backslash (I\cup J)$.  
  We know that $N\cup I $ and $N\cup J$ are  spanning trees in $G$, and so $N$ is a forest.
  The coefficient in front of $m$  is obtained by setting all Schwinger parameters in $N$ to  zero, and taking the coefficient of all remaining parameters $M$. In other words,   the  coefficient of $m$ is exactly
  $$\Psi^{I,J}_{G\backslash M/N} $$
  which by Proposition \ref{E} is given by 
    $$  \det(\mathcal{E}_{G\backslash M/N}( I)) \det(\mathcal{E}_{G\backslash M/N}(J)) \in \{0,\pm1\}\ .$$

    This vanishes if $I$ and $J$ are not spanning trees in $G\backslash M/N$.
The only information of $m$ which remains in $G\backslash M/N$ is which  end points of edges of $I$ and $J$ belong to the same tree of $N$.
  Thus every monomial which gives the same set partition of $V(I \cup J)$ has the same sign, and every monomial corresponds to some such set partition. 
\end{proof}

\begin{example}\label{middle}
One simple example is the case where $I$ and $J$ each consist of a single edge with a common vertex $v$.  Let $u$ and $w$ be the other two end points.  The only set partition with a nonzero coefficient in this case is $\{v\}\{u,w\}$.  Graphically, if
\[
  G = \raisebox{-1cm}{\includegraphics{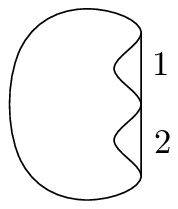}} \qquad \text{ then } \qquad \Psi_G^{1,2} = \pm\raisebox{-1cm}{\includegraphics{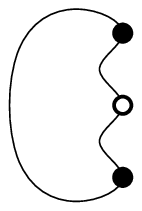}}
\]

\end{example}

\begin{definition} Let $I, J$ be two subsets of edges of $G$ with $|I|=|J|$ and let $P$ be a partition of $V(I\cup J)$. Let $I_P$ (resp. $J_P$,  $\big(I\cup J\big)_P $) denote the graph obtained from the subgraph $I$ (resp. $J$, $I\cup J$) by identifying vertices which lie in the same partition. If the edges of $G$ are oriented, or are ordered, then the graphs $I_P$, $J_P$, $\big(I\cup J\big)_P$ inherit these extra structures also.  If the vertices of $G$ are ordered then the graphs $I_P$, $J_P$, $\big(I\cup J\big)_P$ inherit this structure by using the first vertex in each part to give the order.
\end{definition}

The proof of the previous proposition shows that $f_k = \Psi^{I,J}_{(I\cup J)_P}$. The coefficient is non-zero precisely when both
$I_P$ and $J_P$ are trees, and  have exactly one vertex of every colour.

\begin{definition}\label{signdef} Let $T$ be a rooted tree, with edges labelled from $\{1,\ldots, n\}$ and  non-root vertices labelled $\{1,\ldots, n\}$.  Choose an orientation on its edges.
We define a number by constructing a bijection
$$\phi: E(T) \rightarrow V(T)$$
and a map $s:E(T) \rightarrow \{\pm 1\}$ as follows.  To  each $e\in E(T)$, $\phi(e)$ associates the vertex meeting $e$ which is furthest from the root, and  $s(e)$
is $+1$ if $\phi(e)$ is the endpoint of the oriented edge $e$, and $-1$ if it is the initial point. Define the sign of the (oriented, numbered) tree $T$ to be:
$$\varepsilon(T) = \mathrm{sgn}(\phi) \prod_{e\in E(T)} s(e)\ .$$
\end{definition}

\begin{prop} Choose an ordering on the edges and vertices of $G$ and an orientation of its edges. Let $I,J$ be subsets of edges of $G$ and $P$ a partition on $V(I\cup J)$. Then
$$\Psi^{I,J}_{(I\cup J)_P} = \varepsilon(I_P) \varepsilon(J_P).$$
\end{prop}
\begin{proof} From the proof of proposition \ref{exists}, we have
    $$ \Psi^{I,J}_{(I\cup J)_P} = \det(\mathcal{E}_{(I\cup J)_P}( I_P)) \det(\mathcal{E}_{(I\cup J)_P}(J_P)).$$
This vanishes if $I_P$ and $J_P$ are not both trees. In the case where they are, an inspection of the matrices $\mathcal{E}_{(I\cup J)_P}( I_P)$ shows that its determinant is exactly
$\varepsilon(I_P)$.  The choice of vertex $v$ to remove in $\Psi^{I,J}_{(I\cup J)_P}$ becomes the choice of roots for $I_P$ and $J_P$ by taking as roots the vertices of $I_P$ and $J_P$ which correspond to the part of $v$.
\end{proof}

As a corollary we can easily understand how the sign changes under simple transformations of $P$.

  Let $I$ and $J$ be sets of edges of $G$ with $|I| = |J|$ and $I \cap J = \emptyset$.  Let
  \[
    \Psi_G^{I,J} = \sum_{k} f_k \Phi_G^{P_k}
  \]
  be as in Proposition \ref{exists}.  Suppose $P_i$ and $P_j$ are set partitions appearing in the sum which agree on $V(I)$.  

Without loss of generality, order the vertices so that all of $V(I)$ comes before $V(J)$.  Then the matrices on the $I$ side are identical for $P_i$ and $P_j$ and we have an order of the parts of $P_i$ and $P_j$ on $J$ (coming from the vertex order on $I$) which gives the vertex order on $J_{P_i}$ and $J_{P_j}$.
%
%
%

\begin{cor}\label{trans}
  Suppose $P_i$ and $P_j$ differ on $V_J$ by a transposition of two vertices in the same tree of $J$. Then
  \[
    f_i = - f_j
  \]
\end{cor}

\begin{proof}
The choice of a root $r$ in $J_P$ determines a choice of root in each tree of $J$ recursively.  To see this, first take each vertex which reduces to $r$ as a root in its tree.  Next, for each tree which under $P$ has a vertex $v$ which is identified with a vertex in an already rooted tree of $J$, take $v$ as the root.  Continue until all trees are rooted.  Since only the vertex of each edge which is furthest from $r$ in $J_P$ contributes to the sign, we can compute the sign tree by tree and multiply, using the order of the vertices given by the order of the parts on $J$.

Let $v$ and $w$ be the two transposed vertices in the statement above.  If neither are a root, then the permutation $\phi$ from Definition \ref{signdef}, taken tree by tree, is composed with a transposition, changing the sign.  If it is not possible to choose the root of $J_P$ so that this occurs, then $v$ and $w$ must be the only vertices in their tree of $J$.  In this case switching the parts of $v$ and $w$ has the same effect as switching the orientation of the edge between them, again changing the sign.
\end{proof}

%
%
%
%
%
%

\begin{cor}\label{treeswitch}
   Suppose $P_i$ and $P_j$ differ on $V_J$ by one vertex having switched parts.  Then
  \[
    f_i = f_j 
  \]
\end{cor}

\begin{proof}
Let $v$ be the vertex which changes part.  Let $F$ be the forest of two trees given by $J$ with the identifications of $P_i$ or $P_j$ except that $v$ is not identified with its part-mates.  As in the previous corollary, we can compute the sign at the level of $F$.  By taking the root of $J_{P_i}$ and $J_{P_j}$ to be in the tree of $F$ which does not contain $v$ we get $v$ to be the root of the other tree of $F$.  Thus, by the definition of $\phi$, $v$ does not contribute to the sign.
\end{proof}

%

\begin{example}
  Let 
  \[
    G = \raisebox{-1.5cm}{\includegraphics{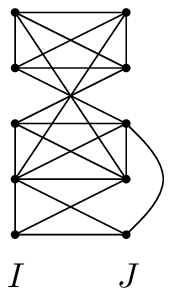}}
  \]
  with the three vertical edges on the left being $I$ and the three vertical edges on the left being $J$.  Consider the two set partitions
  \[
    \includegraphics{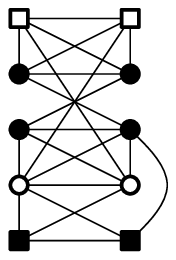} \qquad \includegraphics{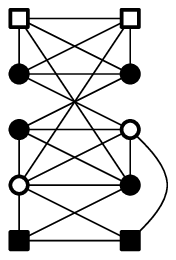}
  \]
  which agree on the ends of $I$ and differ by a transposition on the ends of $J$.  Orient the edges of $I$ and $J$ downwards.  The graphs $J_{P_1}$, $J_{P_2}$ are
    \[
    \includegraphics{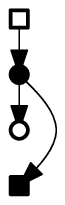} \qquad \includegraphics{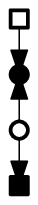}
    \]
    so with the filled square as the removed vertex
    \[
    \mathcal{E}_{(I \cup J)_{P_1}}(I_{P_1}) = \begin{bmatrix} 1 & -1 & 0 \\ 0 & 1 & -1 \\ 0 & 1 & 0 \end{bmatrix}
    \text{ and } 
    \mathcal{E}_{(I \cup J)_{P_2}}(I_{P_2}) = \begin{bmatrix} 1 & -1 & 0 \\ 0 & -1 & 1 \\ 0 & 0 & 1 \end{bmatrix} 
    \]  Calculate $\det(\mathcal{E}_{(I \cup J)_{P_1}}(I_{P_1})) = -1$ and $\det(\mathcal{E}_{(I \cup J)_{P_2}}(I_{P_2}))=1$ as expected from Proposition \ref{trans}.  
\end{example}

\begin{example}
  Let 
  \[
    G = \raisebox{-1.5cm}{\includegraphics{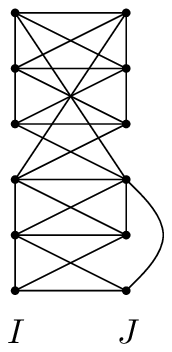}}
  \]
  with the four vertical edges on the left being $I$ and the four vertical edges on the left being $J$.  Consider the two set partitions
  \[
    \includegraphics{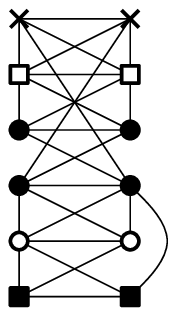} \qquad \includegraphics{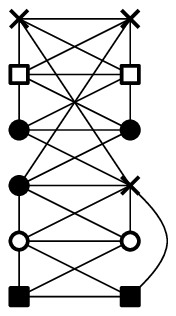}
  \]
  which agree on the ends of $I$ and differ by a single vertex having changed colour on the ends of $J$.  Orient the edges of $I$ and $J$ downwards.
    With the empty square as the removed vertex
    \[
  \mathcal{E}_{(I \cup J)_{P_1}}(I_{P_1}) = \begin{bmatrix}1 & 0 & 0 & 0 \\ 0 & -1 & 0 & 0 \\ 0 & 1 & -1 & 0 \\ 0 & 1 & 0 & -1 \end{bmatrix}
    \text{ and } 
   \mathcal{E}_{(I \cup J)_{P_2}}(I_{P_2}) = \begin{bmatrix}1 & 0 & 0 & 0 \\ 0 & -1 & 0 & 0 \\ 1 & 0 & -1 & 0 \\ 1 & 0 & 0 & -1 \end{bmatrix} 
    \]  Both determinants are $-1$ as expected from Proposition \ref{treeswitch}. 
\end{example}

\section{Identities}\label{sectId}
Spanning forest polynomials give a nice way to look at graph polynomial identities.  To illustrate this we will first recast useful known identities and then prove a new identity which generalizes results of \cite{Brbig}.

We say a graph is \emph{two vertex reducible} if we can remove two vertices of the graph, and the adjoining edges, and the resulting graph is disconnected.

Let $G_1$ and $G_2$ be two graphs.  Let $e_1$ be an edge of $G_1$ and $e_2$ an edge of  $G_2$.  Then define a \emph{two-vertex join} of $G_1$ and $G_2$ to be the graph resulting from identifying $e_1$ and $e_2$ and then cutting the new edge.  Given $e_1$ and $e_2$ there are two ways to do this identification.  However, this ambiguity is of little interest to us because the period of the graph does not depend on it \cite{Sphi4}, nor does the graph polynomial.

\begin{prop}\label{2join} 
Let $G$  be  a two-vertex join as above with $v_1$ and $v_2$ the join vertices. Suppose that edges $i,j,k,l \in G$ are such that  $i,j\in G_1$ and $k,l\in G_2$. Then
\begin{align*}
\Psi_G & = \Phi_{G_1\backslash e_1}^{\{v_1\},\{v_2\}}\Phi_{G_2\backslash e_2} + \Phi_{G_1\backslash e_1}\Phi_{G_2\backslash e_2}^{\{v_1\},\{v_2\}}   \\
\Psi_G^{ij,kl} & =0 \\ 
\Psi_G^{ik,jl} & =\Psi_G^{il,jk}.
\end{align*}
\end{prop}

\begin{proof}
  The first identity holds because every spanning tree of $G$ either connects $v_1$ and $v_2$ on the $G_1$ side or on the $G_2$ side.  In either case the remaining side has a forest of two trees, one connected to $v_1$ the other to $v_2$.  Pairings of such a forest with a spanning tree on the other side give all spanning trees of $G$.

Consider $\Psi^{ij,kl}$.  
Any monomial appearing in this polynomial is a monomial in $\Psi_{G/\{i,j\}}$ and in $\Psi_{G/\{k,l\}}$.
Thus any spanning forest polynomial appearing in the expansion of $\Psi^{ij,kl}$ comes from a partition with exactly three parts and all three parts are represented among the end points of $i$ and $j$ as well as among the end points of $k$ and $l$.  This means there are three trees which connect $G_1$ to $G_2$ in $G$.  However $G_1$ and $G_2$ join at only two vertices.  This is a contradiction and so there are no such monomials.

Then $\Psi^{ik,jl} = \Psi^{il,jk}$ by the Pl\"ucker identity \cite{Brbig}
\[
  \Psi_G^{ij,kl} - \Psi_G^{ik,jl} + \Psi_G^{il,jk} = 0.
\]
\end{proof}

%
%

\begin{prop}\label{transfer}
  Let $u$, $v$, and $w$ be vertices of a graph $G$.  Then
  \begin{align*}
    & \Phi_G^{\{u,v,w\}} \Phi_G^{\{u\},\{v\},\{w\}}
    \\ & = \Phi_G^{\{u,v\},\{w\}}\Phi_G^{\{u,w\},\{v\}} + \Phi_G^{\{u,v\},\{w\}}\Phi_G^{\{u\},\{v,w\}} + \Phi_G^{\{u,w\},\{v\}}\Phi_G^{\{u\},\{v,w\}}
  \end{align*}
  Graphically,
  \[
  \includegraphics[width=0.95\linewidth]{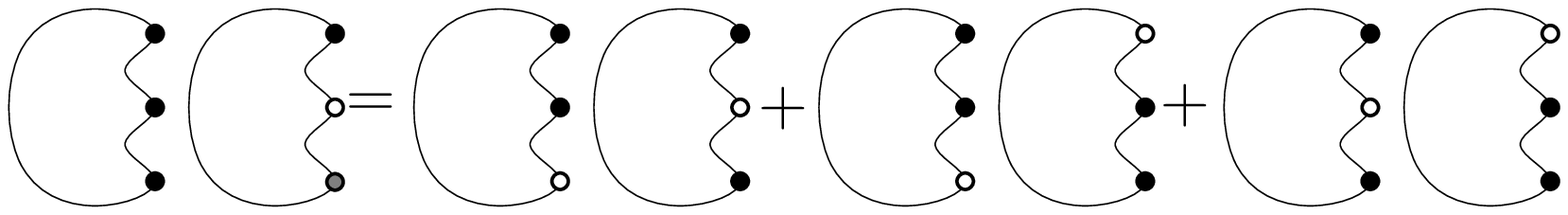}.
  \]
\end{prop}

This identity is essentially the so-called Dodgson identity in this context.  

\begin{proof}
%
%
%
  Use the Dodgson identity, 
   \begin{align*}
    & \det(M_{G}(12,12))\det(M_{G}) \\
    & = \det(M_{G}(1,1))\det(M_{G}(2,2)) - \det(M_{G}(1,2))\det(M_{G}(2,1))
  \end{align*}
  along with Proposition \ref{contract} and Example \ref{middle}.
\end{proof}

The graphical formulation suggests interpreting this proposition as a result about transferring an extra edge from any term in the left hand factor of the left hand side to the right hand factor of the left hand side, thus cutting a spanning tree into two in the left hand factor and joining two of the three trees together in the right hand factor.  The proposition says that transferring an edge in this way results in two spanning forests with exactly two trees in all possible ways.  However carrying this idea through to a proof is delicate as there are many ways the cutting can be done and many ways to build any particular forest of two trees with an extra edge.  It isn't obvious a priori that the counting works out and the authors are not aware of a proof along these lines, though one must surely exist.

\begin{thm} \label{thm3connected}
  Let 
  \[
    G = \raisebox{-1cm}{\includegraphics{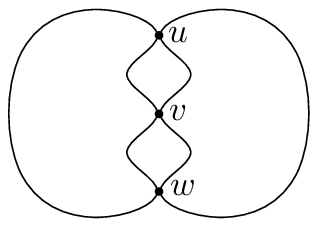}}
  \]
  be a graph which is three vertex reducible at vertices $\{u,v,w\}$.  Let $G'$ and $G''$ be the two halves of $G$ when separated at $u$, $v$, and $w$. Let
  \begin{align*}
    f_1 & = \Phi_{G'}^{\{u\},\{v,w\}} &\hfill  g_1 & = \Phi_{G''}^{\{u\},\{v,w\}}\\
    f_2 & = \Phi_{G'}^{\{v\},\{u,w\}} & \hfill g_2 & = \Phi_{G''}^{\{v\},\{u,w\}}\\
    f_3 & = \Phi_{G'}^{\{w\},\{u,v\}} & \hfill g_3 & = \Phi_{G''}^{\{w\},\{u,v\}}\\
    f & = \Psi_{G'} = \Phi_{G'}^{\{u,v,w\}} & \hfill g & = \Psi_{G''}= \Phi_{G''}^{\{u,v,w\}}  \\
    \widetilde{f} & = \Phi_{G'}^{\{u\},\{v\},\{w\}} & \hfill \widetilde{g} & = \Phi_{G''}^{\{u\},\{v\},\{w\}}
  \end{align*}
  Then
  \begin{align*}
    & f^{\deg(g)+1}g \Psi_G \\
    & = \bigg(f_1'f_2' + f_1'f_3' + f_2'f_3' + f_1'g_2 + f_1'g_3 + f_2'g_1  \\
    & \qquad + f_2'g_3 + f_3'g_1 + f_3'g_2 + g_1g_2 + g_1g_3 + g_2g_3\bigg)\bigg|_{\alpha_1, \alpha_2, \ldots, f\beta_1, f\beta_2, \ldots} \\
  \end{align*}
  where $f_i' = f_ig$, $\alpha_1, \cdots$ are the variables for the edges of $G'$ and $\beta_1, \cdots$ are the variables for the edges of $G''$.
\end{thm}

Note that the piece in parentheses of the expression for $\Psi_G$, 
\[
f_1'f_2' + f_1'f_3' + f_2'f_3' + f_1'g_2 + f_1'g_3 + f_2'g_1
    + f_2'g_3 + f_3'g_1 + f_3'g_2 + g_1g_2 + g_1g_3 + g_2g_3
\]
is itself the graph polynomial of the following graph with the indicated polynomials as edge variables
\[
  \includegraphics{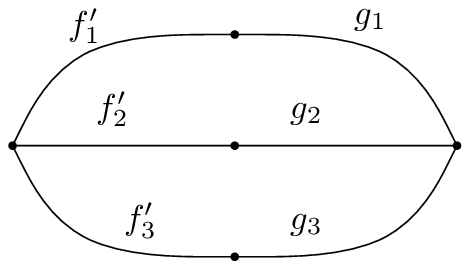}
\]

Graphically, then, we can represent the theorem as
\[
  \includegraphics{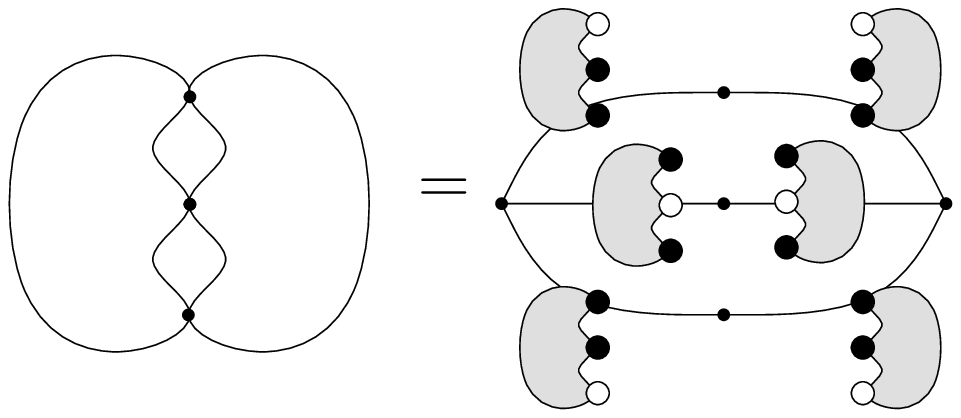}
\]
Note however that this picture does not capture all the details of the theorem as it does not indicate the scalings by powers of 
\[
\includegraphics{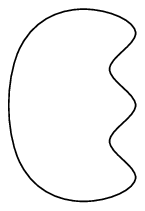} \text{ and } \includegraphics{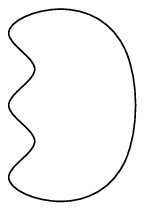}.
\]

\begin{proof}
  Any tree contributing to a term of $\Psi_G$ breaks up uniquely into a spanning forest of $G'$ and a spanning forest of $G''$.  A pair of a spanning forest of $G'$ and a spanning forest of $G''$ give a tree of $G$ precisely when each tree of each forest contains at least one of $u$, $v$, and $w$ and when there is exactly one path between each pair of $u$, $v$, and $w$ using both spanning forests.
Thus
  \[
    \Psi_G = \widetilde{f}g + f_1g_2 + f_1g_3 + f_2g_1 + f_2g_3 + f_3g_1 + f_3g_2 + f\widetilde{g}
  \]
  Let $n = \deg g$.  Note that $\deg g_i = n+1$ and $\deg \widetilde{g} = n+2$.
  Multiplying by $f^{n+1}$ and using Proposition \ref{transfer} we get
  \begin{align*}
    f^{n+1}\Psi_G = & f^{n}(f_1f_2g + f_1f_3g + f_2f_3g) \\
    & + f^{n+1}(f_1g_2 + f_1g_3 + f_2g_1 + f_2g_3 + f_3g_1 + f_3g_2) + f^{n+2}\widetilde{g}
  \end{align*}
  Let $\alpha_1, \ldots, $ be the edges of $G'$ and let $\beta_1, \ldots, $ be the edges of $G''$.  Scale the edges of $G''$ by $f$ giving
  \begin{align*}
   f^{n+1}\Psi_G & = \bigg(f_1f_2g + f_1f_3g + f_2f_3g + f_1g_2 + f_1g_3 + f_2g_1  \\
    & \qquad + f_2g_3 + f_3g_1 + f_3g_2 + \widetilde{g}\bigg)\bigg|_{\alpha_1, \alpha_2, \ldots, f\beta_1, f\beta_2, \ldots} \\
  \end{align*}
  Next multiply by $g$ and use Proposition \ref{transfer}
  \begin{align*}
& f^{n+1}g \Psi_G   \\ & = \bigg(f_1'f_2' + f_1'f_3' + f_2'f_3' + f_1'g_2 + f_1'g_3 + f_2'g_1  \\
    & \qquad + f_2'g_3 + f_3'g_1 + f_3'g_2 + g_1g_2 + g_1g_3 + g_2g_3\bigg)\bigg|_{\alpha_1, \alpha_2, \ldots, f\beta_1, f\beta_2, \ldots} \\
  \end{align*}
  where $f_i' = f_ig$, which is the desired result.
\end{proof}

\section{Weight drop in Feynman graphs}\label{sectHyper}

\subsection{Hyperlogarithms}
Let $\sigma_0=0$, and let $\sigma_1,\ldots, \sigma_N$ denote $N$ distinct points in $\C^*$. Let $\Sigma=\{\sigma_0,\ldots, \sigma_N\}$, and let
$$Y= \C \backslash \Sigma\ .$$
Let $A=\{\x_0,\ldots, \x_N\}$ denote an alphabet with $N+1$ letters, and let $A^*$ denote the free non-commutative monoid on $X$, which consists of the set of all words $w$ in the alphabet $A$ and the empty word $e$. Let $\C\lc A \rc $ denote the ring of non-commutative formal power series in $A$, equipped with the concatenation product. For any element $S\in \C\lc A \rc$, let $S_w$ denote the coefficient of $w$ in $S$, i.e.,  
$$S= \sum_{w\in A^*} S_w w \quad \ , \quad  S_w \in \C\ .$$
Consider the trivial bundle $E=Y \times \C\lc A \rc $ over $Y$, and consider the following one-form on $Y$: 
$$  \Omega(z) =  \Big(\sum_{i=0}^N \x_i {dz \over z-\sigma_i} \Big)$$
Since $d\Omega(z) = \Omega(z) \wedge \Omega(z) =0$ this defines a flat connection on $E$. There is a unique multivalued  section $L:Y \rightarrow E$ which
satisfies:   
\begin{eqnarray}\label{diffeq}
d L(z)&=&\Omega(z) L(z)\ , \\
L(z) & \sim & \exp(\x_0 \log z) \nonumber 
\end{eqnarray}
 where the notation  $L(z) \sim \exp(\x_0 \log z)$ means that there exists a function $h(z)$ holomorphic in the neighbourhood of the origin, such that
 $h(0)=1$ and $L(z) = h(z) \exp(\x_0 \log z)$ for $z$ near $0$.
If $w=\x_{i_1}\ldots \x_{i_r}$ where $i_r \neq 0$, then one can  show that the coefficient  $L_w(z)$ of $L(z)$ is an iterated integral:
$$L_w(z) = \int_{0\leq t_r\leq t_{r-1}\leq \ldots \leq t_1 \leq z} {dt_r \over t_r- \sigma_{i_r}} \ldots{dt_2 \over t_2 - \sigma_{i_2}} {dt_1 \over t_1 -\sigma_{i_1}} $$
for $z\in \mathbb{R}$ in a neighbourhood of $0$.
%
The  equations $(\ref{diffeq})$ are equivalent to the following system of differential equations on the coefficients $L_w(z)$, and determine them uniquely:
\begin{eqnarray}
{\partial \over \partial z} L_{\x_iw}(z)& = & {1\over z-\sigma_i} L_w(z) \hbox{ for } i=0,\ldots, N \nonumber \\
L_{\x_0^n} ( z)& = &  { 1\over n!} \log^n (z) \nonumber \\
L_{w}(z) & \sim & 0 \hbox{ as } z\rightarrow 0 \hbox{ if } w\neq x_0^n  \nonumber 
\end{eqnarray}

\begin{lemma}\label{intlem} We deduce the following indefinite integrals (with constants of integration omitted), where the  denominators  are of degree at most 2 in $z$:
\begin{eqnarray}
&i). & \int {L_w(z) \over z-\sigma_i}dz = L_{\x_i w}(z) \nonumber \\
& ii). & \int{L_w(z) \over (z-\sigma_i)(z-\sigma_j)}dz = { 1\over \sigma_i -\sigma_j}\Big( L_{\x_iw}(z) - L_{\x_jw}(z) \Big)\nonumber \\
& iii). & \int L_{\x_{i_1}\ldots \x_{i_n}}(z) dz = \sum_{k=1}^n (-1)^{k-1}(z-\sigma_{i_k}) L_{\x_{i_{k+1}}\ldots \x_{i_n}}(z)   \nonumber \\
& iv). & \int {L_{\x_i^r \x_j w}(z) \over (z-\sigma_i)^2 } dz  =  {1\over z-\sigma_i} \Big(\sum_{k=1}^r (-1)^{k+1} L_{\x_i^{r-k}\x_jw}\Big) + { 1\over \sigma_i -\sigma_j}\Big( L_{\x_iw}(z) - L_{x_jw}(z) \Big)\nonumber  
\end{eqnarray}
where $i\neq j$, $r\geq 0$,  and $i_1,\ldots, i_n$ are any indices in $\{0,\ldots, n\}$.
\end{lemma}
\begin{proof} (i) and  (ii) follow from the definition of the functions $L_w(z)$ and partial fractions. (iii) and (iv) follow by integration by parts and induction.
\end{proof}
\begin{definition}  Let $\LL$ denote the $\Q$-vector space spanned by the multivalued functions $L_w(z)$, for $w\in A^*$ (which can be shown to be linearly independent).
It is graded by the weight, where the weight $|w|$ of a word $w\in A^*$ is the number of letters in $w$. We write $\LL = \bigoplus_{n\geq 0 } \gr^W_n \LL$, and $W_k\LL = \bigoplus_{0\leq n\leq k} \gr^W_n \LL$.
\end{definition}
It turns out that $\LL$ is closed under multiplication (by the shuffle product formula), and is in fact a graded Hopf algebra.
The various cases of the previous lemma are summarized in the following corollary.
\begin{cor}\label{cordenom} Let $F(z)$ be an element of $\gr^W_k \LL$ of weight $k$, and let $P(z)= a z^2+bz+c$ be a polynomial in $z$ of degree at most 2 with zeros in $\Sigma$. Let $\Delta(P) = b^2-4ac$ denote the discriminant of $P(z)$.   Then
$$\int {F(z) \over P(z)} dz \in  
\begin{cases} {1\over \sqrt{\Delta(P)}} \gr^W_{k+1} \LL & \text{if $\Delta(P)\neq 0$, \qquad (``no weight drop")}
\\
W_{k}  \LL  &\text{if $\Delta(P)= 0$. \qquad (``weight drop")}
\end{cases}
 $$
Note that in the case when a weight drop occurs, the primitive is not necessarily  pure as  there may be mixing of weights.
\end{cor}
\begin{proof} The case $\Delta(P)=0$ corresponds to $(iii)$ and $(iv)$ in the previous lemma, and the remaining cases $(i)$ and $(ii)$ correspond to $\Delta(P)\neq 0$.
\end{proof}

\subsection{Initial integrations} One can use the algebras $\LL$ to integrate out the first few variables in a Feynman integral.
Let $G$ be a primitive-divergent graph, and choose an order on its edges. Consider the residue:
$$I_G = \int _{[0,\infty]^{e_G}} {1 \over \Psi_G^2}\prod_{i=1}^{e_G} d\alpha_i \, \delta(\alpha_{e_G}=1)$$
We can successively integrate out the variables $\alpha_1,\ldots, \alpha_5$ using lemma $\ref{intlem}$. Dropping the $\delta$s from the notation, this gives \cite{Brbig}:
$$I_G^1 = \int {1\over \Psi^{1,1}_{G} \Psi_{G,1} } \prod_{i=2}^{e_G} d\alpha_i$$
$$I_G^2 = \int{ \log \Psi^{1,1}_{G,2} +\log \Psi^{2,2}_{G,1} - \log \Psi^{12,12}_G -\log \Psi_{G,12} \over (\Psi^{1,2}_G)^2} \prod_{i=3}^{e_G} d\alpha_i$$
$$I_G^3 = \int\Big({\Psi^{123,123}_G \log \Psi_G^{123,123} \over    \Psi_G^{12,13}\Psi_G^{12,23}\Psi_G^{13,23} }-{\Psi_{G,123} \log \Psi_{G,123} \over    \Psi^{2,3}_{G,1}\Psi_{G,2}^{1,3}\Psi_{G,3}^{1,2} }+\qquad \qquad \qquad$$
$$\qquad \qquad \qquad 
\underset{\{i,j,k\}}{\sum} {\Psi^i_{G,jk} \log \Psi^i_{G,jk} \over    \Psi_G^{ij,ik}\Psi_{G,j}^{i,k}\Psi_{G,k}^{i,j} } - {\Psi^{ij}_{G,k} \log \Psi^{ij}_{G,k} \over    \Psi_G^{ij,ik}\Psi_G^{ij,jk}\Psi_{G,k}^{i,j} }   \Big) \prod_{i=4}^{e_G} d\alpha_i   $$
where the sum runs over permutations of $\{1,2,3\}$ and so there are 8 terms in the last integral.
Continuing in a similar way and exploiting the many algebraic relations between the polynomials $\Psi^{I,J}_K$ \cite{Brbig}, one verifies that:
\begin{eqnarray}
I_G^4 &=& \int \Big( {A \over \Psi^{12,34}_G\Psi^{13,24}_G} + {B \over \Psi^{14,23}_G\Psi^{13,24}_G} +{C \over \Psi^{12,34}_G\Psi^{13,24}_G} \Big)   \prod_{i=5}^{e_G} d\alpha_i \\
I_G^5 & = & \int{ F\over {}^5\Psi_G(1,2,3,4,5)} \prod_{i=6}^{e_G} d\alpha_i\ .
\end{eqnarray}
where $A,B,C$ are hyperlogarithms of weight 2, and $F$ is a hyperlogarithm of weight $3$ with singularities in $\{\Psi^{I,J}_{G,K}=0\}$ where $I\cup J \cup K= \{1,2,3,4,5\}$.
The 5-invariant ${}^5\Psi(1,2,3,4,5)$ is defined as follows:
\begin{definition} The 5-invariant of any 5 edges  $i,j,k,l,m$  in $G$ is:
$${}^5\Psi_G(i,j,k,l,m) =  \pm \det
\left(
\begin{array}{cc}
 \Psi_{G,m}^{ij,kl} & \Psi_G^{ijm,klm}   \\
\Psi_{G,m}^{ik,jl}  &     \Psi_G^{ikm,jlm} 
\end{array}
\right)
$$
It is well-defined, i.e., permuting $i,j,k,l,m$ in the above only changes the sign of the determinant.
\end{definition}

\subsection{Denominator reduction}
The denominator reduction is an algorithm for computing the denominators at successive stages of integration using Corollary $\ref{cordenom}$.

\begin{definition} Let $G$ be a  primitive-divergent graph and choose an ordering on its set of edges.
 Let $D_5={}^5\Psi_G(1,2,3,4,5)$.  Let $n\geq 5$, and suppose inductively that $D_n$ factorizes into a product of linear factors in $\alpha_{n+1}$, i.e.,
 $D_n = (a\alpha_{n+1}+b) (c\alpha_{n+1}+d)$. Then we define
 $$D_{n+1} = \sqrt{\Delta(D_n)} =\pm( ad-bc)\ ,$$
 where the discriminant $\Delta$ is taken with respect to $\alpha_{n+1}$. 
 A graph $G$ for which the polynomials $D_n$ can be defined for all $n$ is called \emph{denominator-reducible}. If for some $n$, $D_n$ is a perfect square in $\alpha_{n+1}$, then 
 $D_{n+1}$, and all $D_m$ for $m\geq n+1$ are identically $0$. In  this case, we say that $G$ has a \emph{weight drop}. Otherwise, $G$ is \emph{non-weight drop}.
\end{definition}

\begin{remark}The interpretation of $D_n$ as a denominator proves that $D_n$ does not depend on the chosen order of variables up to that point.  We will frequently use the notation
$${}^n\Psi_G(e_1,\ldots, e_n)\qquad n\geq 5, $$
to denote the denominator $D_n$ of the graph $G$ after reducing with respect to the edges $e_1,\ldots, e_n$.

\end{remark}

We have the following rather naive definition of the transcendental weight of a period:
\begin{definition} Let $P,Q$ be polynomials in $\overline{\Q}[x_1,\ldots, x_n]$  and consider an absolutely convergent period of the form:
$$I = \int_{[0,\infty]^n} {P(x_1,\ldots, x_n) \over Q(x_1,\ldots, x_n)} dx_1\ldots dx_n\ .$$
We say that such a  period  has weight at most $n$ if it can be written  as  a sum of  convergent period  integrals as above with at most $n$ integrations.\footnote{This corresponds to the fact that the mixed Hodge structure of a complex of open affine varieties of dimension  at most $n$ has weights contained in $[0,2n]$.}
This defines a filtration on the set of these periods. 
\end{definition}
The definition of weight as given above is compatible with the weight filtration on the elements $L_w(z)$ of $\LL$.
It  is satisfactory for periods of mixed Tate motives, and in particular  should give back the usual notion of weight   for multiple zeta values,
however, in the general case it is a little  simplistic; as remarked earlier,  
a more sophisticated approach to the weight is to view $I$ as a period of a mixed Hodge structure.
It is remarked in \cite{Brbig} that the arguments we give in this paper do in fact prove an analogous result on the weights in the Hodge-theoretic sense.

It follows from the computations of $I^4_G$ above that $I$ can be written as  a $2+e_G-5=e_G-3$ fold integral (each term $A,B,C$  is of weight 2 and can therefore be written as a 2-fold integral), and there remain $e_G-5$ Schwinger parameters to integrate out owing to the $\delta(\alpha_{e_G}-1)$ term. It follows that the weight of $I_G$ is at most $e_G-3$.

\begin{thm}\cite{Brbig} Suppose that $G$ is primitively divergent as above and has a weight drop at the $n^{\mathrm{th}}$ stage of its denominator reduction.
If furthermore $G$  is linearly reducible up to the  $n^{\mathrm{th}}$ point, then the weight of $I_G$ is at most $e_G-4$. 
\end{thm}
The  linear reducibility condition guarantees that the integrands  in the integration process  can indeed be written as hyperlogarithms (i.e., they are multivalued functions on a Zariski open subset of projective space and have global unipotent monodromy). The previous calculations make it  clear that every graph $G$ is linearly reducible up to the $5^{\mathrm{th}}$ stage. Sufficient conditions for linear reducibility  are given in \cite{Brbig}. 
 \begin{remark} If $G$ is linearly reducible and has no weight drop, then the expected transcendental weight of $I_G$ is $e_G-3$.
 \end{remark}The purpose of this paper is 
 to investigate the \emph{combinatorial} conditions under which a weight drop (defined by the vanishing of a denominator $D_n$) occurs.

\subsection{Weight-preserving operations}

Let
\[
  G = \raisebox{-1.3cm}{\includegraphics{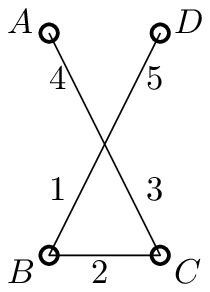}}
\]
where the circled vertices indicate where the explicitly drawn edges attach to the rest of the graph, which is left undrawn to avoid clutter.  Let $K$ be the rest of the graph, $K = G\backslash \{1,2,3,4,5\}$.
Let 
\[
 H=G\backslash \{2\} = \raisebox{-1.3cm}{\includegraphics{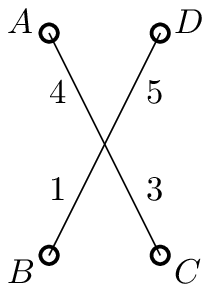}}
\]

\begin{prop}\label{dt5prop}
\begin{align*}
{}^5\Psi_G(1,2,3,4,5)& = \pm\Psi^{13,45}_{H} \Psi^{4,5}_{H,13} \\
& = \pm\left(\Phi^{\{A,B\},\{C,D\}}_K - \Phi^{\{A,C\},\{B,D\}}_K\right)\Phi^{\{A,D\},\{B\},\{C\}}_K
\end{align*}
\end{prop}

\begin{proof} 
Since $123$ forms a triangle, we have
$${}^5\Psi_G(1,2,3,4,5)= \pm\Psi_G^{123,245}\Psi^{14,35}_{G,2}.$$  Drawing $K$ as a blob and using $\cap$ to indicate the polynomial formed of terms common to each argument with signs as in Proposition \ref{exists}, we have
\begin{align*}
  \Psi_{G,2}^{14,35} & = \raisebox{-1cm}{\includegraphics{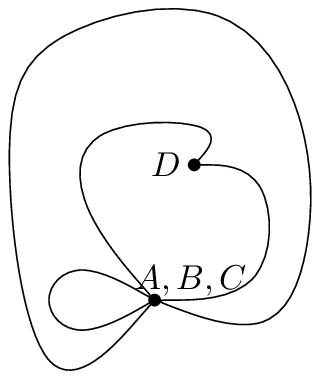}}\cap \raisebox{-1cm}{\includegraphics{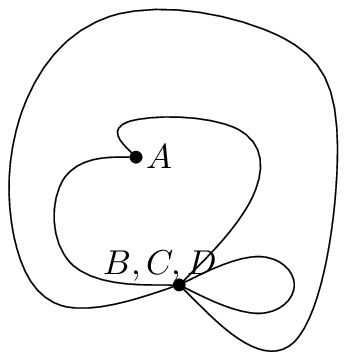}} \\ &= \pm\Phi^{\{A,D\},\{B\},\{C\}}_K \\
\end{align*}
\begin{align*}
  \Psi_{G}^{123,245} & = \raisebox{-1cm}{\includegraphics{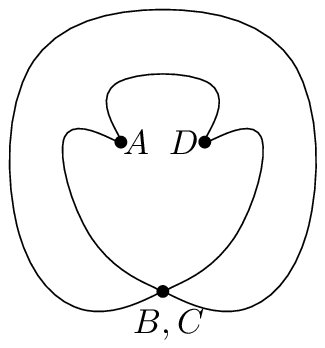}}\cap \raisebox{-1cm}{\includegraphics{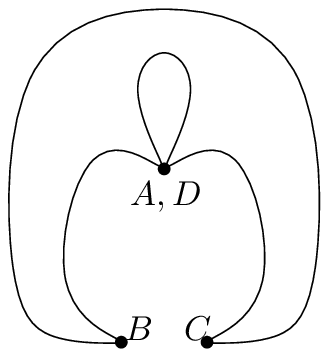}} \\ &= \pm\left(\Phi^{\{A,B\},\{C,D\}}_K - \Phi^{\{A,C\},\{B,D\}}_K\right)
\end{align*}
$\Psi_{H,13}^{4,5}$ gives the same intersection of blobs as $\Psi_{G,2}^{14,35}$ and $\Psi_H^{13,45}$ gives the same intersection of blobs as $\Psi_{G}^{123,245}$.  The result follows.
\end{proof}

Consider a `double triangle':
\[
  G' = \raisebox{-1.3cm}{\includegraphics{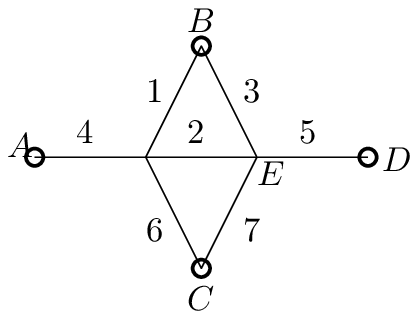}}
\]
Let $K$ again be the rest of the graph, $K = G'\backslash \{1,2,3,4,5,6,7\}$.  Note that $H = \left(G'\backslash \{3, 7\}/2\right)|_{6 \leftrightarrow 3}$.
\begin{prop}\label{dt7prop}
The denominator $D_7$ after reducing $G'$ with respect to the seven edges $1-7$ indicated above, is given by
\begin{align*}D_7(G') & = \pm \Psi^{13,45}_{H} \Psi^{4,5}_{H,13} \\ & = \pm\left(\Phi_K^{\{A,B\},\{C,D\}} - \Phi_K^{\{A,C\},\{B,D\}}\right)\Phi_K^{\{A,D\},\{B\},\{C\}} \end{align*}
\end{prop}
\begin{proof}By Proposition \ref{dt5prop} applied to edges $1,3,2,4,6$ we know that
\[
  {}^5\Psi_{G'}(1,2,3,4,6) = \pm\Phi_{K\cup\{5,7\}}^{\{A,C\},\{B\},\{E\}}\left( \Phi_{K\cup\{5,7\}}^{\{A,B\},\{C,E\}} - \Phi_{K\cup\{5,7\}}^{\{A,E\},\{B,C\}}\right)
\]
Notice that in the partition $\{A,C\}, \{B\},\{E\}$ the two ends of edge $7$ are in different parts.  Thus, by Proposition \ref{propdiffparts}
\[
 \Phi_{K\cup\{5,7\}}^{\{A,C\},\{B\},\{E\}} = \alpha_7 \Phi_{K\cup 5}^{\{A,C\},\{B\},\{E\}} .
\] 
This removes one term in the discriminant so we can easily apply a denominator reduction with respect to the edge $7$. We deduce that
\begin{align*}
  {}^6\Psi_{G'}(1,2,3,4,6,7) & = \pm\Phi_{K\cup 5}^{\{A,C\},\{B\},\{E\}}\Phi_{K \cup \{5,7\}/7}^{\{A,B\},\{C\}} \\
  & = \pm \raisebox{-1.3cm}{\includegraphics{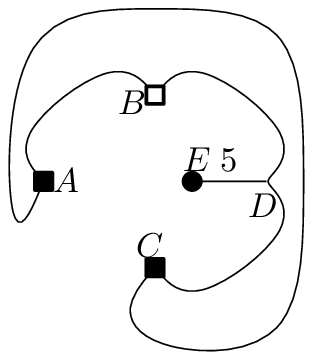}}\raisebox{-1.3cm}{\includegraphics{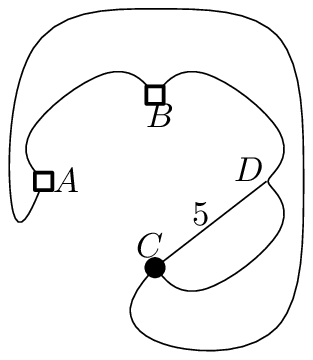}}
\end{align*}
From the pictures we can read off the contractions and deletions of edge $5$ and deduce that the reduction with respect to edge $5$ is
\begin{align*}
  & {}^7\Psi_{G'}(1,2,3,4,5,6,7) \\
  & = \pm\left(\Phi_K^{\{A,C\},\{B\}}\Phi_K^{\{A,B\},\{C\},\{D\}} - \Phi_K^{\{A,C\},\{B\},\{D\}}\Phi_K^{\{A,B\},\{C\}}\right)
\end{align*}
But this is itself a five-invariant, ${}^5\Psi_G(1,2,3,4,5)$, expanded as 
\[
  \Psi_G^{145,235}\Psi_{G,5}^{12,34} - \Psi_{G,5}^{14,23}\Psi_G^{125,345}
\]
 where $G$ is as in the previous proposition.  Applying Proposition \ref{dt5prop} to rewrite this $5$-invariant completes the proof.
\end{proof}

\begin{thm}\label{split}
Let $G$ and $G'$ be obtained, as above, by splitting a triangle. Suppose that $G'$ is linearly reducible with respect to a set of edges
 $\{1,\ldots, 7\}\cup S$ (in that order), where $S\subset G'\backslash \{1,\ldots, 7\}$. Then $G$ is linearly reducible with respect to $\{1,\ldots, 5\}\cup S$ and has a weight drop  if and only if $G'$ has a weight drop.
\end{thm}

\begin{proof}
It follows from the two previous propositions that:
$${}^5\Psi_G(1,2,3,4,5) =\pm {}^7\Psi_{G'}(1,2,3,4,5,6,7)\ .$$
\end{proof}

Note that $G'$  
 is always linearly reducible with respect to $\{1,2,3,4,5,6,7\}$ (the case $S=\emptyset$), because every 5-edge minor of $G'$  has either a triangle or a 3-valent vertex, and so $G'$ 
  contains no non-trivial 5-invariants.

By deleting edges $6$ and $7$ in $G'$ we get a special case of the double triangle where a single triangle with two three-valent vertices is contracted to two three-valent vertices connected by an edge.  If we also consider the two remaining edges adjacent to $B$, then by similar arguments $7$ and $5$ integrations respectively give the same denominator.  By deleting edge $4$ we get a special case with a three-valent vertex in the double triangle contracting to a single triangle with a three-valent vertex.  If we also consider one more edge adjacent to $C$ then by similar arguments $7$ and $5$ integrations respectively again give the same denominator.

\subsection{Families of weight drop graphs}

The first family of weight drop graphs is already well known.

\begin{prop}\label{2VR}
Let $G$ be two vertex reducible.  Then $G$ has a weight drop.
\end{prop}

\begin{proof}  Write $G$ as the 2-vertex join of two graphs $G_1$ and $G_2$. Number the edges of $G$ in any way so that edges $1,2$ lie in $G_1$ and edges $3, 4$ lie in $G_2$.
Then by Proposition \ref{2join}, $\Psi^{12,34}_G= 0$ and $\Psi_G^{13,24}=\Psi_G^{14,23}$. At the fourth stage of integration ($I^4_G$ above) the denominator reduces to 
$$\Psi_G^{13,24}\Psi_G^{14,23} = (\Psi_G^{13,24})^2\ .  $$
Thus we have a weight drop. 
\end{proof}

The same argument shows that any graph with a double edge and more than 4 edges in total has weight drop.

The first family of weight drop graphs which goes beyond 2-vertex reducible graphs was observed empirically by one of us (KY) and independently by Oliver Schnetz, who later also found a proof.  The most beautiful proof, also observed by Oliver Schnetz, is in terms of the material of this section.


\begin{example}
  Every graph of the form 
  \[
  \includegraphics{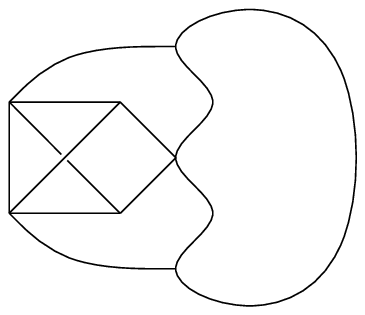}
  \]
  has a weight drop.

  To see this consider the pair of triangles marked below by heavy edges
  \[
  \includegraphics{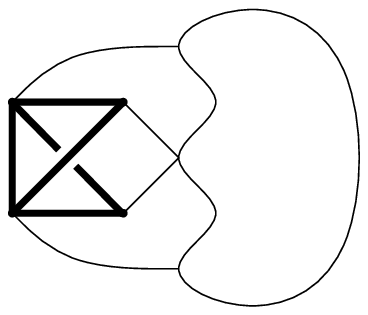}
  \]
  By Theorem \ref{split} this graph has weight drop iff 
  \[
  \includegraphics{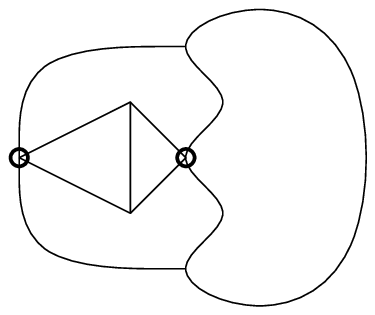}
  \]
  does.  The latter is two vertex reducible at the marked vertices and hence by Proposition \ref{2VR} has weight drop.
\end{example}

Further families can be built along these lines by repeated double triangles.  To make a more systematic search for weight preserving operations and weight drop families we can use the 3-vertex join result as in the next section.

\subsection{Operations on  3-connected graphs}\label{sect3vw}
Let $G$ be a 3-connected graph.   We can write $G=L\cup_3 R$ as the join of two graphs $L$ and $R$ along three distinguished vertices $v_1,v_2,v_3$ (below left).
By Theorem \ref{thm3connected}  there is a universal formula for the graph polynomial of $G$ in terms of Dodgson polynomials. Thus for any fixed left-hand side $L$,
we can consider the graph $\widetilde{L}$ obtained by joining a vertex $v$ to $v_1, v_2,v_3$ (right):

\[
\includegraphics{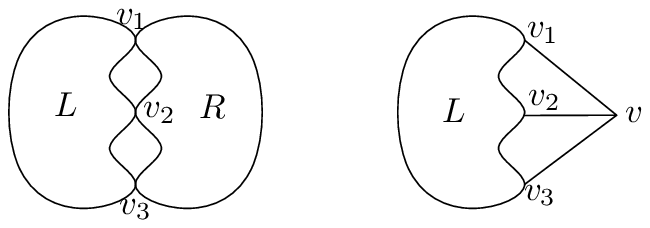}
\]

Suppose that $L$ has at least $5$ edges, and let $x,y,z$ denote the Schwinger parameters of the edges $\{v,v_1\},\{v,v_2\},\{v,v_3\}$ respectively.  Suppose that $\widetilde{L}$ is denominator
reducible. 
Then we define 
$$ \rho_L(x,y,z) = {}^{e_L}\Psi_{\widetilde{L}}(\alpha_1,\ldots, \alpha_{e_L}) \in \Q[x,y,z]$$
 where the edges of $L$ are numbered $1,\ldots, e_L$. By Theorem \ref{thm3connected}, the polynomial  $\rho_L$ computes the general shape of the denominator reduction of any graph $G=L\cup_3 R$ after reducing out all
 the  edges in $L$.
 
 \begin{prop} Suppose that $G$ is the 3-vertex join of $L$ and $R$. Let
$$x_R= \Phi_{R}^{\{1\},\{2,3\}}\ , \ y_R= \Phi_{R}^{\{2\},\{1,3\}}\ , \ z_R= \Phi_{R}^{\{1,2\},\{3\}}\ .$$
  We have shown that $x_Ry_R+x_Rz_R+y_Rz_R= \Psi_R \Phi_R^{\{1,2,3\}}$ (Proposition \ref{transfer}), and
  $$x_R+y_R= \Psi^{3}_{\widetilde{R},12} \ , \ x_R+z_R =\Psi^{2}_{\widetilde{R},13} \ , \  y_R+z_R =\Psi^{1}_{\widetilde{R},23}\ . $$
  If $G$ is denominator reducible, then Theorem $\ref{thm3connected}$ implies that: 
 $${}^{e_L} \Psi_G(\alpha_1,\ldots, \alpha_{e_L}) = (\Psi_R)^{2-\deg \rho_L} \rho_L ( x_R,y_R,z_R)\ .$$
 It follows in particular that if $\rho_L$ is of degree $0$  or is  a perfect square, then $G$ has  a weight drop.
 \end{prop}
There are a limited number of possibilities for the polynomials $\rho_L$. This gives rise to families of weight-preserving operations as follows.
\begin{cor} \label{cor3vw} Suppose that $\rho_{L_1}=\rho_{L_2}$, and  let
 $G=L_1\cup_3 R$   and $G_2= L_2 \cup_3 R$ for any $R$.  Then the denominator reductions of $G_1$ and $G_2$ are the same after reducing out all the edges in $L_1$ and $L_2$:
 $$ {}^{e_{L_1}} \Psi_{G_1}(\alpha_1,\ldots, \alpha_{e_{L_1}}) ={}^{e_{L_2}} \Psi_{G_2}(\alpha_1,\ldots, \alpha_{e_{L_2}})\ . $$
 Thus $G_1$ has a weight drop after reducing with respect to the edges $ E(L_1)\cup S$, where $S\subset E(R)$,  if and only if $G_2$ 
 has a weight drop after reducing with respect to $E(L_2) \cup S$.
\end{cor}
We begin a classification of graphs $L$ and compute their polynomials $\rho_L$ as follows. 
In the following diagrams, the white vertices 
are $v_1,v_2,v_3$ from top to bottom, and the polynomials $\rho_L$ are indicated underneath. At 5 edges, there are only two possibilities for $L$ which have neither a double edge nor a two-valent vertex (these are the simple 5-local minors in the terminology of \cite{Brbig}):
\begin{center}
\fcolorbox{white}{white}{
  \begin{picture}(214,105) (228,-50)
    \SetWidth{1.0}
    \SetColor{Black}
    \Vertex(235,14){4}
    \SetWidth{1.0}
    \Line(235,14)(289,52)
    \Line(235,14)(289,14)
    \Line(235,14)(289,-25)
    \Arc(435,52)(4,270,630)
    \Arc(435,14)(4,270,630)
    \Arc(435,-25)(4,270,630)
    \SetWidth{1.0}
    \Vertex(381,38){4}
    \Vertex(381,-9){4}
    \SetWidth{1.0}
    \Line(381,38)(435,52)
    \Line(381,38)(381,-9)
    \Line(381,38)(435,14)
    \Line(435,14)(381,-9)
    \Line(381,-9)(435,-25)
    \Text(230,-49)[lb]{{\Black{$y(xy+xz+yz)$}}}
    \Text(405,-49)[lb]{{\Black{$y$}}}
    \Text(230,50)[lb]{{\Black{$5_1$}}}
    \Text(370,50)[lb]{{\Black{$5_2$}}}
   \Arc(289,52)(4,270,630)
    \Arc(289,14)(4,270,630)
    \Arc(289,-25)(4,270,630)
    \Line(289,52)(289,-25)
  \end{picture}
}
\end{center}
Note that  graphs $5_1$ and $5_2$  have a split triangle (resp. split 3-valent vertex),  so  if they occur in a graph, we can, as noted after Theorem \ref{split}, reduce to a smaller graph, except in trivial cases where the extra edges are not available.
At 6 edges, there are exactly six such 6-local minors:
\begin{center}
\fcolorbox{white}{white}{
  \begin{picture}(312,205) (72,-39)
    \SetWidth{1.0}
    \SetColor{Black}
    \Vertex(76,130){3}
    \SetWidth{1.0}
    \Arc(131,162)(4,270,630)
    \Arc(131,130)(4,270,630)
    \Arc(131,99)(4,270,630)
    \Arc(249,162)(4,270,630)
    \Arc(249,130)(4,270,630)
    \Arc(249,99)(4,270,630)
    \Arc(368,162)(4,270,630)
    \Arc(368,130)(4,270,630)
    \Arc(368,99)(4,270,630)
    \Arc(131,51)(4,270,630)
    \Arc(131,20)(4,270,630)
    \Arc(249,51)(4,270,630)
    \Arc(131,-12)(4,270,630)
    \Arc(368,51)(4,270,630)
    \Arc(368,20)(4,270,630)
    \Arc(368,-12)(4,270,630)
    \Arc(249,-12)(4,270,630)
    \Arc(249,20)(4,270,630)
    \SetWidth{1.0}
    \Vertex(202,146){3}
    \Vertex(202,114){3}
    \Vertex(320,146){3}
    \Vertex(320,114){3}
    \Vertex(76,36){3}
    \Vertex(76,4){3}
    \Line(76,130)(131,162)
    \Line(76,130)(131,130)
    \Line(76,130)(131,99)
    \Line(131,162)(131,99)
    \Arc(108,130.5)(39.003,-53.865,53.865)
    \Line(202,146)(249,162)
    \Line(202,146)(249,130)
    \Line(202,146)(249,99)
    \Line(202,114)(249,162)
    \Line(202,114)(249,99)
    \Line(320,146)(320,115)
    \Line(320,146)(367,162)
    \Line(320,114)(367,130)
    \Line(320,114)(367,99)
    \Line(320,146)(367,130)
    \Arc(341.433,130)(40.57,-50.936,50.936)
    \Vertex(202,36){3}
    \Vertex(202,4){3}
    \Vertex(320,44){3}
    \Vertex(320,-4){3}
    \Vertex(336,20){3}
    \Line(320,44)(320,-4)
    \Line(320,44)(336,20)
    \Line(336,20)(320,-4)
    \Line(320,-4)(367,-12)
    \Line(336,20)(367,20)
    \Line(320,44)(367,52)
    \Line(76,36)(76,5)
    \Line(76,4)(131,-12)
    \Line(76,36)(131,20)
    \Line(76,4)(131,20)
    \Line(76,36)(131,52)
    \Line(131,52)(131,20)
    \Line(202,36)(202,5)
    \Line(202,36)(249,52)
    \Line(202,36)(249,20)
    \Line(202,36)(249,-12)
    \Line(202,4)(249,52)
    \Line(202,4)(249,-12)
    \Line(202,114)(249,130)
    \Text(66,155)[lb]{{\Black{$6_1$}}}
    \Text(190,155)[lb]{{\Black{$6_2$}}}
    \Text(308,155)[lb]{{\Black{$6_3$}}}
   \Text(66,50)[lb]{{\Black{$6_4$}}}
    \Text(190,50)[lb]{{\Black{$6_5$}}}
    \Text(308,50)[lb]{{\Black{$6_6$}}}
    \Text(70,75)[lb]{{\Black{$(xy+yz+xz)^2$}}}
    \Text(190,75)[lb]{{\Black{$xy+yz+xz$}}}
    \Text(340,77)[lb]{{\Black{$xz$}}}
    \Text(84,-35)[lb]{{\Black{$(x+y)y$}}}
    \Text(220,-35)[lb]{{\Black{$xz$}}}
    \Text(340,-35)[lb]{{\Black{$1$}}}
  \end{picture}
}
\end{center}
The graphs $6_1$ and $6_6$ have weight drops, and we obtain the first identity:
$\rho_{6_3}= \rho_{6_5}$. However,  most of the above graphs (except for  $6_2$ and $6_6$)  contain a double-triangle or split 3-vertex and do not tell us anything new.  From now on,  we only consider  graphs which do not contain a double triangle, rather than giving the complete list.  It turns out that at 7 edges, we obtain polynomials $\rho_L$ which have already appeared
above, except for the following graph:
\begin{center}
\fcolorbox{white}{white}{
  \begin{picture}(70,95) (219,-150)
    \SetWidth{2.0}
    \SetColor{Black}
    \Arc(287,-63)(4,270,630)
    \Arc(287,-96)(4,270,630)
    \Arc(287,-130)(4,270,630)
    \SetWidth{1.0}
    \Vertex(220,-77){4}
    \Vertex(220,-116){4}
    \Line(220,-77)(287,-63)
    \Line(287,-63)(287,-96)
    \Line(220,-77)(287,-96)
    \Line(220,-77)(287,-130)
    \Line(220,-116)(287,-130)
    \Line(220,-116)(287,-96)
    \Line(220,-116)(287,-63)
 \Text(200,-60)[lb]{{\Black{$7$}}}
 \Text(200,-150)[lb]{{\Black{$(x+y)(xy+xz+yz)$}}}
  \end{picture}
}
\end{center}
Similarly, at 8 edges, we get a new identity for the graphs: 
\begin{center}
\fcolorbox{white}{white}{
  \begin{picture}(266,100) (102,-140)
    \SetWidth{2.0}
    \SetColor{Black}
    \Arc(178,-47)(4,270,630)
    \Arc(178,-83)(4,270,630)
    \Arc(178,-118)(4,270,630)
    \SetWidth{1.0}
    \Vertex(107,-47){4}
    \Vertex(107,-83){4}
    \SetWidth{2.0}
    \Arc(363,-47)(4,270,630)
    \Arc(363,-83)(4,270,630)
    \Arc(363,-118)(4,270,630)
    \SetWidth{1.0}
    \Vertex(292,-47){4}
    \Vertex(292,-83){4}
    \Vertex(292,-118){4}
    \Vertex(107,-118){4}
    \Line(292,-47)(363,-47)
    \Line(107,-47)(178,-47)
    \Line(107,-118)(178,-118)
    \Line(107,-47)(107,-118)
    \Line(107,-47)(178,-83)
    \Line(107,-83)(178,-118)
    \Line(292,-47)(363,-83)
    \Line(107,-118)(178,-83)
    \Line(107,-83)(178,-47)
    \Line(292,-47)(292,-118)
    \Line(292,-118)(363,-118)
    \Line(292,-47)(363,-118)
    \Line(292,-118)(363,-47)
    \Line(292,-83)(363,-83)
    \Text(120,-140)[lb]{{\Black{$y(x+z)$}}}
    \Text(312,-140)[lb]{{\Black{$y(x+z)$}}}
 \Text(80,-50)[lb]{{\Black{$8_a$}}}
    \Text(265,-50)[lb]{{\Black{$8_b$}}}
     \end{picture}
}
\end{center}
neither of which is amenable to a double-triangle type reduction. 
We conclude that $\rho_{8_a}= \rho_{8_b}$.
We conclude with a few more examples to illustrate  the general principle: 
\begin{center}
\fcolorbox{white}{white}{
  \begin{picture}(331,114) (20,-82)
    \SetWidth{2.0}
    \SetColor{Black}
    \Arc(222,15)(4,270,630)
    \Arc(222,-24)(4,270,630)
    \Arc(222,-63)(4,270,630)
    \SetWidth{1.0}
    \Vertex(283,15){4}
    \SetWidth{2.0}
    \Arc(346,-63)(4,270,630)
    \Arc(346,-24)(4,270,630)
    \Arc(346,15)(4,270,630)
    \Arc(89,15)(4,270,630)
    \Arc(89,-24)(4,270,630)
    \Arc(89,-63)(4,270,630)
    \SetWidth{1.0}
    \Vertex(26,15){4}
    \Vertex(26,-63){4}
    \Vertex(49,-9){4}
    \Vertex(49,-39){4}
    \Vertex(159,15){4}
    \Vertex(159,-63){4}
    \Vertex(182,-9){4}
    \Vertex(182,-39){4}
    \Vertex(283,-63){4}
    \Vertex(307,-9){4}
    \Vertex(307,-39){4}
    \Line(26,15)(89,15)
    \Line(26,15)(50,-9)
    \Line(49,-9)(49,-40)
    \Line(49,-39)(25,-63)
    \Line(26,15)(26,-63)
    \Line(49,-9)(88,-24)
    \Line(49,-39)(88,-24)
    \Line(49,-39)(88,-63)
    \Line(26,-63)(89,-63)
    \Line(159,15)(159,-63)
    \Line(159,15)(183,-9)
    \Line(182,-9)(182,-40)
    \Line(182,-39)(158,-63)
    \Line(159,15)(222,15)
    \Line(283,15)(283,-63)
    \Line(307,-9)(307,-40)
    \Line(307,-39)(283,-63)
    \Line(283,-63)(346,-63)
    \Line(307,-39)(346,-63)
    \Line(307,-9)(346,-24)
    \Line(346,-24)(307,-39)
    \Line(283,15)(346,15)
    \Line(346,15)(307,-9)
    \Line(283,15)(307,-9)
    \Line(182,-9)(221,-24)
    \Line(222,-24)(183,-39)
    \Line(182,-39)(221,-63)
    \Line(222,-63)(159,-63)
    \Line(222,15)(222,-24)
    \Text(42,-82)[lb]{{\Black{$y+z$}}}
    \Text(182,-82)[lb]{{\Black{$z^2$}}}
    \Text(283,-82)[lb]{{\Black{$xy+xz+yz$}}}
    \Text(6,14)[lb]{{\Black{$9$}}}
    \Text(136,14)[lb]{{\Black{$10_a$}}}
    \Text(262,14)[lb]{{\Black{$10_b$}}}
  \end{picture}
}
\end{center}
Thus $\rho_{10_a}$ has weight drop, and we get an identity $\rho_{10_b} = \rho_{6_2}$.
One can continue generating larger and larger graphs $L$ (provided they are denominator-reducible) and compute their polynomials $\rho_L$, giving rise to more and more complicated identities of the form $\rho_{L_1}=\rho_{L_2}$.
\subsection{Examples of 3-connected operations}
The previous discussion enables one to prove results about graphs which are inaccessible by double-triangle type arguments.
\begin{cor}\label{K34andG}
The following   graphs have a  weight drop:
\begin{center}
\fcolorbox{white}{white}{
  \begin{picture}(267,96) (102,-160)
    \SetWidth{2.0}
    \SetColor{Black}
    \Arc(143,-140)(4,270,630)
    \SetWidth{1.0}
    \Vertex(107,-90){4}
    \Vertex(107,-119){4}
    \Text(107,-155)[lb]{{\Black{$K_{3,4}
$}}}
    \Vertex(178,-90){4}
    \Vertex(178,-126){4}
    \SetWidth{2.0}
    \Arc(143,-105)(4,270,630)
    \Arc(143,-69)(4,270,630)
    \SetWidth{1.0}
    \Line(143,-69)(178,-126)
    \Line(143,-69)(178,-91)
    \Line(178,-90)(143,-140)
    \Line(178,-126)(143,-140)
    \Line(178,-90)(143,-105)
    \Line(143,-105)(178,-127)
    \Line(107,-90)(143,-69)
    \Line(107,-90)(143,-105)
    \Line(143,-105)(108,-119)
    \Line(107,-119)(143,-140)
    \Line(143,-69)(108,-119)
    \Line(107,-90)(143,-140)
    \Vertex(271,-69){4}
    \SetWidth{2.0}
    \Arc(328,-140)(4,270,630)
    \Arc(328,-105)(4,270,630)
    \Arc(328,-69)(4,270,630)
    \SetWidth{1.0}
    \Vertex(271,-140){4}
    \Vertex(293,-90){4}
    \Vertex(293,-119){4}
    \Line(271,-69)(271,-140)
    \Line(293,-90)(293,-119)
    \Line(293,-119)(271,-140)
    \Line(271,-140)(328,-140)
    \Line(293,-119)(328,-140)
    \Line(293,-90)(328,-105)
    \Line(328,-105)(293,-119)
    \Line(271,-69)(328,-69)
    \Line(328,-69)(293,-90)
    \Line(271,-69)(293,-90)
    \Text(306,-155)[lb]{{\Black{$G$}}}
    \Vertex(364,-90){4}
    \Vertex(364,-126){4}
    \Line(328,-69)(364,-90)
    \Line(364,-90)(328,-105)
    \Line(328,-105)(363,-127)
    \Line(364,-126)(328,-140)
    \Line(328,-69)(363,-126)
    \Line(364,-90)(328,-140)
  \end{picture}
}
\end{center}
The same result holds more generally for  $K_{3,4} \cup_3 R$ or $G\cup_3 R$ where $R$ is any  3-vertex reducible graph connected to the 3 white vertices.
\end{cor}
\begin{proof}
The graph $K_{3,4}=6_2\cup_3 6_2$ on the left is denominator reducible and has a weight drop by direct calculation  \cite{Brbig}.
 Since $\rho_{10_b}=\rho_{6_2}$ it follows that  $G=10_b\cup_3 6_2$ also has  a weight drop by 
  Corollary  \ref{cor3vw}. The last statement follows from a similar argument, after noting (again by direct computation) that $\rho_{\widetilde{K_{3,4}}}=0$.
  \end{proof}

Now we can consider the graphs obtained by gluing $8_a$ and $8_b$ together.
There are three possibilities which preserve the symmetries of the $\rho$ polynomials: $8_a\cup 8_a$, $8_a\cup 8_b$, $8_b\cup 8_b$
\begin{equation}\label{aandb}
\includegraphics[width=.85\linewidth]{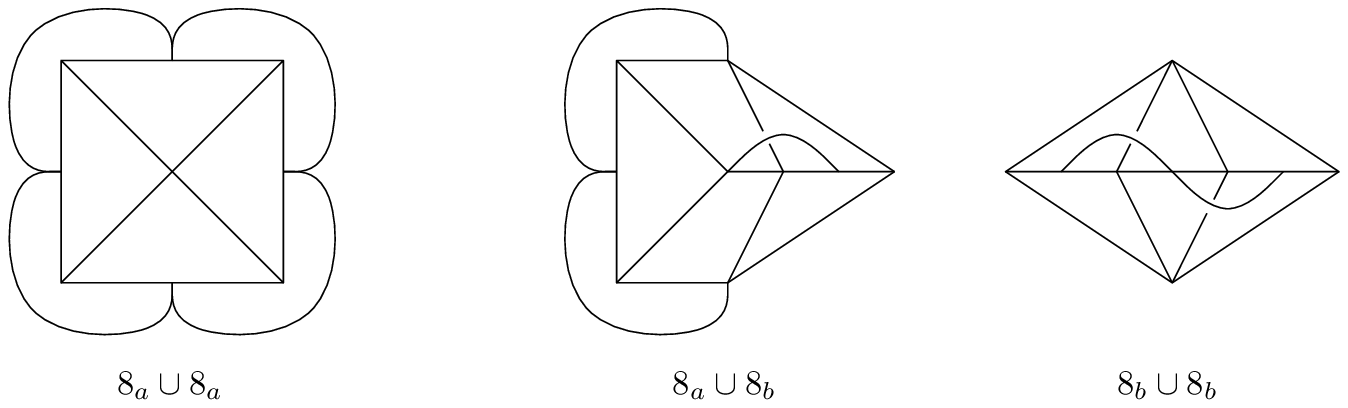}
\end{equation}
Arguing as above we conclude that all 3 have the same weight. 

\begin{remark}
  It is known that two graphs $G_1$, $G_2$, which, when completed by adjoining a new vertex to all 3-valent vertices give rise to the same graph, necessarily
have the same period.  O. Schnetz has shown this nicely in \cite{Sphi4}. In particular, they have the same weight. This should imply that $G_1$ has a weight drop in the denominator sense if and only if $G_2$ does. Can one find a combinatorial proof of this fact?
\end{remark}

\begin{remark} We obtained a list of 3-vertex connected operations on graphs simply by calculating their universal polynomials $\rho$. Is it possible to do the same for graphs with higher
degrees of vertex connectivity? In other words, for every $n\geq 3$,  is there a  finite list of  `right-hand sides' $R_1,\ldots, R_{N_n}$ such that if a property holds for all $L\cup_n R_i$ for $1 \leq i \leq N_n$
then it holds for all graphs $L\cup_n R$? In the 3-vertex connected case we have shown that $N_3=1$. It would be interesting to draw up a list of 4-vertex connected weight-preserving operations, of which the 
triangle splitting operation is one example.
\end{remark}

The results above are almost sufficient to explain all known weight-drops.  Of the graphs up to 8 loops which have been calculated to be weight drop, the application of Theorem \ref{split} and Proposition \ref{2VR} without any further identities explains all but seven of the graphs.  One of these can be explained immediately by planar duality.  Two more of them are $K_{3,4}$ and $G$ from Corollary \ref{K34andG}.  Three more are the graphs of \eqref{aandb} which all must have the same weight. The remaining graph is
\begin{equation}\label{1badgraph}
  \includegraphics{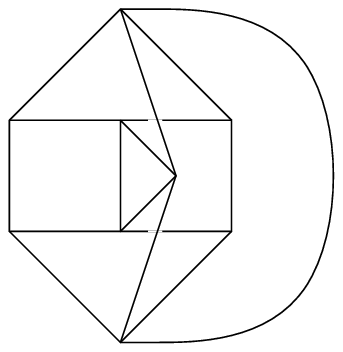}
\end{equation}
So the methods of this paper suffice to prove all known weight drops except for two.  Fortunately, these graphs are amenable to the denominator reduction algorithm.  Alternately, if we further allow ourselves Schnetz' completion and twist operations, which preserve the period \cite{Sphi4}, but which are not understood in this language of graph polynomials, then the graphs of \eqref{1badgraph} and \eqref{aandb} must all have the same value as $G$ from Corollary \ref{K34andG}, giving weight drop for all of them.

\bibliographystyle{plain}
\bibliography{main}

\begin{thebibliography}{1}

\bibitem{bek}
Spencer Bloch, H\'el\`ene Esnault, and Dirk Kreimer.
\newblock On motives associated to graph polynomials.
\newblock {\em Commun. Math. Phys.}, 267:181--225, 2006.
\newblock arXiv:math/0510011v1 [math.AG].

\bibitem{Brbig}
Francis Brown.
\newblock On the periods of some {F}eynman integrals.
\newblock arXiv:0910.0114.

\bibitem{Ko-Za}
M.~Kontsevich and D.~Zagier.
\newblock Periods.
\newblock In {\em Mathematics Unlimited--2001 and Beyond}, pages 771--808.
  Springer, 2001.

\bibitem{Sphi4}
Oliver Schnetz.
\newblock Quantum periods: A census of $\phi^4$-transcendentals.
\newblock arXiv:0801.2856.

\end{thebibliography}

\end{document}